\documentclass[a4paper,12pt, twoside]{article}

\usepackage[margin=1in]{geometry} 
\usepackage{amsmath, amsthm, slashed} 
\usepackage{authblk} 
\usepackage{xparse} 
\usepackage{amssymb} 
\usepackage{enumerate} 
\usepackage{tikz}
\usetikzlibrary{decorations.pathmorphing,calc,intersections}
\usepackage[margin=1in]{geometry}
\usepackage{accents} 
\usepackage{tensor} 
\usepackage{bbm} 
\usepackage{mathrsfs} 
\usepackage{comment} 
\usepackage[normalem]{ulem} 
\usepackage{wasysym} 
\usepackage{relsize} 
\usepackage{upgreek}

\newtheorem{innercustomgeneric}{\customgenericname}
\providecommand{\customgenericname}{} 
\newtheorem*{theorem*}{Theorem} 
\newtheorem{proposition}{Proposition}[section]
\newtheorem*{proposition*}{Proposition}

\newtheorem*{corollary*}{Corollary}

\newtheorem*{lemma*}{Lemma*}

\newtheorem*{remark*}{Remark}

\newtheorem*{definition*}{Definition}
\newtheorem*{conjecture*}{Conjecture}
\newtheorem*{question*}{Question}

\title{On the link between the Maxwell and linearised Einstein equations on Schwarzschild}
\author{Thomas Johnson}
\affil{Imperial College London}

\newcommand{\reals}{\mathbb{R}}

\newcommand{\exd}{\textnormal{d}}
\newcommand{\pt}{\partial_t}

\newcommand{\mcals}{\mathcal{S}}
\newcommand{\slg}{\textsl{g}}

\newcommand{\schst}{\left(\mcals, \textsl{g}_M\right)}
\newcommand{\smsymtwocovSch}{\Gamma(S^2T^*\mcals)}
\newcommand{\Boxgsch}{\Box_{\textsl{g}_M}}
\newcommand{\Riemgsch}{\textnormal{Riem}_{\slg_M}}
\newcommand{\divergsch}{\textnormal{div}_{\slg_M}}
\newcommand{\trgsch}{\textnormal{tr}_{\slg_M}}

\newcommand{\stkout}[1]{\ifmmode\text{\sout{\ensuremath{#1}}}\else\sout{#1}\fi}
\newcommand{\qxi}{\widetilde{\xi}}

\newcommand{\twosphere}{\mathbb{S}^2}

\newcommand{\qg}{\widetilde{g}_M}

\newcommand{\qj}{\widetilde{j}}
\newcommand{\sj}{\slashed{j}}

\newcommand{\omegab}{\underline{\omega}}

\newcommand{\qtr}{\textnormal{tr}_{\qg}}
\newcommand{\qhd}{\widetilde{\star}}

\newcommand{\qexd}{\widetilde{\textnormal{d}}}
\newcommand{\qn}{\widetilde{\nabla}}
\newcommand{\qdiv}{\widetilde{\delta}}

\newcommand{\qbox}{\widetilde{\Box}}

\newcommand{\str}{\slashed{\textnormal{tr}}}

\newcommand{\astrosunhat}{\widehat{{\astrosun}}}
\newcommand{\even}[1]{{#1}_{\textnormal{e}}}
\newcommand{\odd}[1]{{#1}_{\textnormal{o}}}

\newcommand{\sdiv}{\slashed{\textnormal{div}}}

\newcommand{\sn}{\slashed{\nabla}}
\newcommand{\sdso}{\slashed{\mathcal{D}}_1^{\star}}
\newcommand{\sdst}{\slashed{\mathcal{D}}_2^{\star}}
\newcommand{\sdo}{\slashed{\mathcal{D}}_1}
\newcommand{\sdt}{\slashed{\mathcal{D}}_2}
\newcommand{\slap}{\slashed{\Delta}}

\newcommand{\zslapinv}[1]{\slap^{-{#1}}_{\mathfrak{Z}}}
\newcommand{\qA}{\widetilde{A}}
\newcommand{\sA}{\slashed{A}}
\newcommand{\qZ}{\widetilde{Z}}

\newcommand{\qP}{\widetilde{P}}

\newcommand{\qtau}{\widetilde{\tau}}

\newcommand{\qeta}{\widetilde{\eta}}

\newcommand{\qp}{\widetilde{p}}

\DeclareDocumentCommand\norm{m m o o o} {{\mathbb{#1}}_{\IfNoValueF{#4}{\mathsmaller{#4}}}^{\IfNoValueF{#3}{#3}}[#2]{\IfNoValueF{#5}{({#5})}}}

\newcommand{\qgamma}{\widetilde{\gamma}}
\newcommand{\mgamma}{\stkout{\gamma}}
\newcommand{\sgamma}{\slashed{\gamma}}

\newcommand{\qastrosunhat}{\widehat{\widetilde{\astrosun}}}

\newcommand{\opmu}{(1+\mu)}

\newcommand{\qm}{\mcalq}
\newcommand{\mcalq}{\mathcal{Q}}

\newcommand{\qhatgamma}{\widehat{\widetilde{\gamma}}}
\newcommand{\qtrgamma}{\qtr\qgamma}
\newcommand{\shatgamma}{\widehat{\slashed{\gamma}}}
\newcommand{\strgamma}{\str\sgamma}
\newcommand{\qzeta}{\widetilde{\zeta}}

\begin{document}

\maketitle

\begin{abstract}
In this short note we shall demonstrate that given a smooth solution $\gamma$ to the linearised Einstein equations on Schwarzschild which is supported on the $l\geq 2$ spherical harmonics and expressed relative to a transverse and traceless gauge then one can construct from it a smooth solution to the sourced Maxwell equations expressed relative to a generalised Lorentz gauge. Here the Maxwell current is constructed from those gauge-invariant combinations of the components of $\gamma$ which are determined by solutions to the Regge--Wheeler and Zerilli equations. The result thus provides an elegant link between the spin 1 and spin 2 equations on Schwarzschild.
\end{abstract}

\tableofcontents

\section{Introduction and remarks}

The remarkable similarity (see e.g. \cite{Waldbook}) between the linearised Einstein equations in a transverse traceless gauge and the Maxwell equations in a Lorentz gauge is a celebrated feature of the theory of linearised gravity (see e.g. \cite{Keinstein}). The purpose of this note is to give a more precise version of this correspondence on the Schwarzschild spacetime. That is to say, we will show that one can in fact construct a solution to the latter from a solution to the former.

A rough version of this result will be given in section \ref{Roughversionoftheresult}. Beforehand however we discuss some further remarks regarding particular analytical consequences of the result as well as consequences relating to the study of linearised gravity on Schwarzschild. Said remarks are discussed in sections \ref{FurtherremarksI:onmovingfromalgebratoanalysis} and \ref{FurtherremarksII:ontherelationtothelinearstabilityofSchwarzschildinaharmonicgauge}.

\subsection{Rough version of the result}\label{Roughversionoftheresult}

Let $\schst$ denote the Schwarzschild spacetime of mass $M>0$. Then\footnote{Here $\smsymtwocovSch$ denotes the space of smooth, symmetric, 2-covariant tensors on $\mcals$.} $\gamma\in\smsymtwocovSch$ is said (cf. \cite{Waldbook}) to be a smooth solution to the linearised Einstein equations in the transverse traceless gauge on $\schst$ if the following system of equations hold true:
\begin{equation}\label{leett}
\begin{aligned}
\Boxgsch\gamma-2\Riemgsch\cdot\gamma^{\sharp\sharp}&=0,\\
\divergsch\gamma&=0,\\
\trgsch\gamma&=0.
\end{aligned}
\end{equation}
Here $\Boxgsch, \divergsch$ and $\trgsch$ are respectively the wave, divergence and trace operator associated to $\slg_M$ whilst $\Riemgsch$ is the Riemann tensor of $\slg_M$ and we have defined the contraction $(\Riemgsch\cdot\gamma^{\sharp\sharp})_{\alpha\beta}:=(\Riemgsch)_{\gamma\alpha\beta\delta}\gamma^{\gamma\delta}$ with indices raised as standard by $\slg_M$. Then the result is the following:

\begin{proposition}\label{propintro}
Let $\gamma$ be as above. Assume moreover that $\gamma$ is supported on the $l\geq 2$ spherical harmonics. Then there exist smooth 1-forms $A=A(\gamma)$ and $j=j(\gamma)$ along with a smooth function $L=L(\gamma)$ such that
\begin{align*}
\Boxgsch A&=-j-\exd L,\\
\divergsch A&=-L,\\
\divergsch j&=0.
\end{align*}
Moreover the quantities $j$ and $L$ are decoupled from $A$ in that they are determined by two scalars $\Phi=\Phi(\gamma)$ and $\Psi=\Psi(\gamma)$ each of which respectively satisfy the (decoupled) Regge--Wheeler and Zerilli equations.

In particular, $A$ is a smooth solution to the sourced Maxwell equations when expressed in a generalised Lorentz gauge.
\end{proposition} 
We emphasize that the notation $A=A(\gamma)$ etc. is to denote that the relevant quantities are constructed (explicitly) from $\gamma$. See section \ref{Proofofresult} for the precise relations. See there also for what it means for a tensor to be supported on the $l\geq2$ spherical harmonics. This restriction on the angular frequencies is to avoid complications provided by linearised Kerr solutions to \eqref{leett} (see e.g. \cite{Jlinstabschwarz}). We also note that it is necessary to introduce the gauge function $L$ so that the continuity equation holds for $j$.


\subsection{Further remarks I: From algebra to analysis}\label{FurtherremarksI:onmovingfromalgebratoanalysis}

It is of the authors opinion that the algebraic content of the above proposition is sufficiently elegant to be worthy of record. However one can exploit it directly to obtain quantitative \emph{decay} estimates on sufficiently regular solutions to the tensorial system \eqref{leett}\footnote{We note that \eqref{leett} admits a well-posed Cauchy problem with the gauge conditions propagating under evolution by the wave part of the equation.}. The reason why one might wish to establish such estimates is discussed in the next section.

Indeed a consequence of the result is that one can express the solution $\gamma$ in terms of five scalars which collectively satisfy a hierarchical system of scalar wave equations. This system can then be analysed (hierarchically) using the techniques developed by Dafermos--Rodnianski in \cite{DRlecturenotes}, \cite{DRredshift} and \cite{DRrpmethod} for studying the scalar wave equation on Schwarzschild. This ``scalarisation'' of \eqref{leett} is presented in the appendix to this note and arises from the following considerations.

The bottom level of the hierarchy are the previously mentioned decoupled scalars $\Phi$ and $\Psi$. The extraction of these quantities is motivated by their gauge-invariance (\cite{RWrweqn}, \cite{Zzeqn} and \cite{Moncrieflinstab}). Here gauge refers to the infinitesimal diffeomorphism symmetry of linearised gravity. The second level of the hierarchy then consists of two scalars determined from the gauge-invariant part of the Maxwell potential $A(\gamma)$. Here gauge now refers to the $U(1)$ gauge symmetry of electromagnetism. It is well known (see e.g. \cite{COSschwarz}) that in the presence of a Maxwell current such quantities decouple into the inhomogeneous wave equation described by the Fackerell--Ipser equation with inhomogeneity determined from the current.  The final level of the hierarchy then consists of one scalar determined now from the (remaining) gauge-dependent part of the Maxwell potential. That these quantities satisfy inhomogeneous (but decoupled) wave equations follows upon closer inspection of the Maxwell equations if one allows the invariant quantities to act now as source terms. 

It is important to note that the techniques of \cite{DRlecturenotes}-\cite{DRrpmethod} cannot be used to study the asymptotics of solutions to \eqref{leett} directly. Indeed a geometric treatment using covariant derivatives as ``multipliers'' would lead at best to non-coercive energies due to the Lorentzian character of the associated inner product. Moreover the potential term is not positive-definite. On the other hand expressing \eqref{leett} relative to a frame leads to a coupled system of \emph{linear} equations and the global analysis of such systems is not particularly susceptible to energy methods. We thus take the point of view that the main new difficulty in estimating solutions to \eqref{leett} is in gaining a sufficient understanding of the underlying algebraic structure so as to make current techniques applicable. The subsequent analysis is then left to the interested reader.

\subsection{Further remarks II: On the linear stability of Schwarzschild in a harmonic gauge}\label{FurtherremarksII:ontherelationtothelinearstabilityofSchwarzschildinaharmonicgauge}

The reason one might be interested in obtaining the decay estimates discussed in the previous section is due to the fact that \eqref{leett} arises as the linearisation around Schwarzschild of the Einstein vacuum equations when expressed in a harmonic gauge (see e.g. \cite{Jlinstabschwarz} -- note that the traceless condition comes from imposing a unimodular condition in the nonlinear theory). A statement of non-linear stability for Schwarzschild in this gauge thus relies fundamentally, at least with current techniques, on establishing sufficiently good global estimates in the linear theory. It is in fact in this context that the author originally discovered, during the course of his PhD, the main result of this note. Indeed the 1-form $A$ was isolated precisely because it determines the ``gauge-dependent'' part of $\gamma$, it being well known from the pioneering work of Regge--Wheeler \cite{RWrweqn} and later Zerilli \cite{Zzeqn} that the gauge-invariant part is determined by solutions to the Regge--Wheeler and Zerilli equations (corresponding in Proposition \ref{propintro} to the scalars $\Phi$ and $\Psi$). Observe then that the restriction on the angular frequencies becomes entirely natural in view of the fact that the $l=0,1$ modes are described completely by linearised Kerr solutions and infinitesimal changes in the centre of mass and linear momentum.

Unfortunately however the estimates one obtains by proceeding as described in the previous section do not appear to fall under the remit of being ``sufficiently good''. In particular the decay rate obtained for a suitable \emph{energy} on the radiation field $r\gamma$ with $r$ an area radius function is not integrable in ``time'' in the wave zone. Moreover, and perhaps more fundamentally, this energy loses derivatives. The reason for both these losses when compared with the scalar wave equation is that one is estimating a hierarchical system of \emph{linear} wave equations with the source terms  exhibiting ``bad\footnote{A model (near the photon sphere) two tier hierarchy would be $\Boxgsch\phi=\pt\psi$ with $\Boxgsch\psi=0$, $t$ a suitable time function, and the regularity one wants to propagate is $\phi$ being at the level of $\psi$.}'' behaviour near the photon sphere and null infinity (both in terms of $r$ and the derivative appearing on the source. The structure of the hierarchy in question also does not appear to allow these deficiencies to be overcome via means of some re-normalisation procedure.\footnote{Remark that the improved results of \cite{AAGbetterdecay} for decay of the scalar wave equation are needed to even ascend the hierarchy so as to obtain the weak decay estimates for $\gamma$ in question.}

These observations suggest that harmonic gauge is perhaps not the optimal gauge in which to consider the nonlinear stability problem for Schwarzschild. This in turn motivated our search for a \emph{generalised} harmonic gauge and indeed in \cite{Jlinstabschwarz} we were able to identify such a gauge in which both energy decay and energy boundedness holds at the desired level for the radiation field of sufficiently regular solutions. As it happens this choice of gauge arose as a modification of the generalised harmonic gauge also employed by the author in \cite{Johnsonlinstabschwarzold}, the modification allowing estimates on the radiation field to be obtained, and the identification of this latter gauge can be directly related to Proposition \ref{propintro}. Indeed, in this gauge the associated 1-form $A$ actually satisfies the source free Maxwell equations and can in fact be eliminated globally by exploiting residual gauge freedom so as to set the initial data of $A$ to be trivial.\footnote{Remark that the vanishing of $A$ actually defines the Regge--Wheeler gauge of \cite{RWrweqn}. A particular insight of \cite{Johnsonlinstabschwarzold} was thus the realisation that the Regge--Wheeler gauge can in fact be realised as a generalised harmonic gauge.}

Of course it is also possible that the ``scalarisation'' outlined in the previous section is not the most efficient way to analyse solutions of \eqref{leett}. 
Indeed, we note the recent paper \cite{HHVlinstabkerr} which obtains in particular, via microlocal techniques, decay estimates for solutions to the ``larger'' system
\begin{equation}\label{poo}
\begin{aligned}
\Boxgsch\gamma-2\Riemgsch\cdot\gamma^{\sharp\sharp}&=0,\\
\divergsch\gamma-\frac{1}{2}\exd \trgsch\gamma&=0
\end{aligned}
\end{equation}
without resorting to a full ``algebraic'' decomposition.\footnote{Note however that the estimates still lose derivatives.} Yet it would appear that the algebraic insights of section \ref{FurtherremarksI:onmovingfromalgebratoanalysis} could be employed even in this approach. To see this, we first observe that an important ingredient in \cite{HHVlinstabkerr} is obtaining a statement of mode stability for pure gauge solutions to \eqref{poo}. These are solutions of the form $\gamma=\mathcal{L}_V\slg_M$ for a vector field $V$ solving $\Boxgsch V=0$. However this statement is relatively easy to establish for the system \eqref{leett}, which differs from \eqref{poo} only by a residual gauge transformation, by exploiting the insights of the previous section. Indeed we first observe that pure gauge solutions to the system \eqref{leett} must in addition satisfy $\divergsch V=0$ and are thus determined by solutions to the source free Maxwell equations expressed relative to a Lorentz gauge. Mode stability then follows from applying known results for the scalar wave equation on Schwarzschild combined with the argument of the previous section which ``scalarises'' the now source free Maxwell system into a two tier hierarchy of scalar waves.


\subsection{Acknowledgements}

The author thanks Gustav Holzegel for many helpful comments regarding the manuscript. The author in addition gratefully acknowledges support through ERC Consolidator grant 772249. 

\section{Proof of result}\label{Proofofresult}

We now prove Proposition \ref{propintro}. We shall focus in the proof on the derivation only on the Schwarzschild exterior.

 Now the proof will require decomposing the system \eqref{leett} relative to the spherical topology of Schwarzschild as this allows the structure of the tensorial system to be better understood. Whilst there are many ways of doing this let us fix for definiteness the Schwarzschild-star foliation of the exterior by 2-spheres as in \cite{Jlinstabschwarz}. Then in the exterior patch
\begin{align*}
\mcals_e\cong\reals^2_{t^*, r}\times \twosphere
\end{align*}
which then allows the operation of projection of spacetime tensors, objects, operators etc. onto their parts ``tangent'' and ``non-tangent'' to the spheres. This formalism is developed concretely in section 3 and appendix A of \cite{Jlinstabschwarz} and we shall employ in freely in this note for the sake of brevity. In particular, we treat these sections of \cite{Jlinstabschwarz} as a companion piece to this note.\footnote{It is there also that one will find what it means for the tensor $\gamma$ to be supported on the $l\geq 2$ spherical harmonics.}

With this in mind, to state the result precisely we first must project $\gamma$ onto the mixture of $\qm$-tensors and $S$-tensors $\qhatgamma,\qtrgamma, \mgamma,\shatgamma$ and $\strgamma$ as defined in \cite{Jlinstabschwarz}. We then introduce the derived quantities
\begin{align*}
\widetilde{\tau}&=\qhatgamma+\frac{1}{2}\qg\qtrgamma-\qn\astrosun\qzeta\\
\qeta&=\odd{\mgamma}-r^2\qexd\left(r^{-2}\odd{\shatgamma}\right)\\
\sigma&=\strgamma-2\slap\even{\shatgamma}-\frac{4}{r}\qzeta_{\qP}
\end{align*}
with 
\begin{align*}
\qzeta=\even{\mgamma}-r^2\qn\left(r^{-2}\even{\shatgamma}\right).
\end{align*}
Here $\odd{\mgamma}, \even{\shatgamma}$ and $\odd{\shatgamma}$ and the components appearing in the Hodge-type decomposition of $\mgamma$ and $\shatgamma$ of section 3.3 in \cite{Jlinstabschwarz}. They are unique and well defined since $\gamma$ is supported on the $l\geq 2$ spherical harmonics.

\begin{proposition}\label{mainprop}
Let $\gamma$ be as in Proposition \ref{propintro} and define the smooth 1-forms
\begin{align*}
A&=\qzeta+\sdso(\even{\shatgamma}, \odd{\shatgamma}),\\
j&=\frac{2}{r}\qtau_{\qP}-\frac{1}{2}\qn\sigma-\frac{1}{r}\qn r\,\sigma-\frac{1}{2}\sdso(r\sigma, -4\qeta_{\qP})
\end{align*}
along with the smooth function
\begin{align*}
L=\frac{1}{2}\sigma.
\end{align*}
Then the following system of equations hold true:
\begin{align*}
\Boxgsch A&=-j-\exd L,\\
\divergsch A&=-L,\\
\divergsch j&=0.
\end{align*}
\end{proposition}
\begin{proof}
It follows from Proposition A.1 of \cite{Jlinstabschwarz} that the projected quantities associated to $\gamma$ satisfy the system
	\begin{align*}
	\qbox\qhatgamma+\slap\qhatgamma+\frac{2}{r}\qn_{\qP}\qhatgamma-\frac{2}{r}\qn r\qastrosunhat\sdiv\mgamma-\frac{2}{r^2}\qn r\qastrosunhat\qhatgamma_{\qP}-\frac{2\mu}{r^2}\qhatgamma-\frac{1}{r^2}\qn r\qastrosunhat\qn r\big(\qtrgamma-\strgamma\big)&=0,\\
	\qbox\qtrgamma+\slap\qtrgamma+\frac{2}{r}\qn_{\qP}\qtrgamma-\frac{4}{r}\sdiv\mgamma_{\qP}+\frac{4}{r^2}\qgamma_{\qP\qP}-\frac{2}{r^2}(1-2\mu)\big(\qtrgamma-\strgamma\big)&=0,\\
	\qbox\mgamma+\slap\mgamma+\frac{2}{r}\qn_{\qP}\mgamma+\frac{2}{r}\sn\astrosun\qgamma_{\qP}-\frac{2}{r}\qn r\astrosun\sdiv\sgamma-\frac{3}{r^2}\qn r\astrosun\mgamma_{\qP}-\frac{1}{r^2}(1-2\mu)\mgamma&=0,\\
	\qbox\shatgamma+\slap\shatgamma+\frac{2}{r}\qn_{\qP}\shatgamma+\frac{2}{r}\sn\astrosunhat\mgamma_{\qP}-\frac{2}{r^2}\opmu\shatgamma&=0,\\
	\qbox\strgamma+\slap\strgamma+\frac{2}{r}\qn_{\qP}\strgamma+\frac{4}{r}\sdiv\mgamma_{\qP}+\frac{4}{r^2}\qhatgamma_{\qP\qP}+\frac{2}{r^2}(1-2\mu)\big(\qtrgamma-\strgamma\big)&=0,\\
	-\qdiv\qhatgamma+\frac{1}{2}\qn\qtrgamma+\sdiv\mgamma+\frac{2}{r}\qhatgamma_{\qP}+\frac{1}{r}\qn r\big(\qtrgamma-\strgamma\big)&=0,\\
	-\qdiv\mgamma+\sdiv\shatgamma+\frac{1}{2}\sn\strgamma+\frac{3}{r}\mgamma_{\qP}&=0,\\
	\qtrgamma+\strgamma&=0.
	\end{align*}
From this we then derive
\begin{align*}
\qbox\qzeta+\slap\qzeta+\frac{2}{r}\qn_{\qP}\qzeta-\frac{2}{r^2}\qn r\text{ }\qzeta_{\qp}&=\frac{2}{r}\qn r\text{ }  \slap\even{\shatgamma}-\frac{2}{r}\qtau_{\qp}+\frac{1}{r}\qexd r\text{ }\sigma,\\
\qbox\even{\shatgamma}+\slap\even{\shatgamma}&=-\frac{2}{r}\qzeta_{\qp},\\
\qbox\odd{\shatgamma}+\slap\odd{\shatgamma}&=-\frac{2}{r}\qeta_{\qP},\\
\qdiv\qzeta-\slap\even{\shatgamma}-\frac{2}{r}\qzeta_P&=L.
\end{align*}
Therefore
\begin{align}
\qbox\qA+\slap\qA+\frac{2}{r}\qn_{\qP}\qA-\frac{2}{r^2}\qn r\text{ }\qA_{\qP}-\frac{2}{r}\qn r\text{ }  \sdiv{\sA}&=-\qj-\qn L,\label{cock1}\\
\qbox{\sA}+\slap{\sA}+\frac{2}{r}\qn_{\qP}\sA+\frac{2}{r}\sn\qA_{\qP}-\frac{1}{r^2}\qn r\cdot\qn r\sA&=-{\sj}-\sn L\label{cock2},\\
\qdiv\qA-\sdiv{\sA}-\frac{2}{r}\qA_{\qP}&=L\label{cock3}.
\end{align}
This yields the first two equations on $A$. Finally for the divergence equation on $j$ we observe from (57)-(59) on page 49 of \cite{Jlinstabschwarz} that the quantities $\qtau, \qeta$ and $\sigma$ must satisfy
\begin{align*}
\qdiv\qtau+\frac{1}{2}\qn\sigma=\qtr\qtau=\qdiv\qeta=0
\end{align*}
from which the proposition follows.
\end{proof}

Here in the last step of the proof we have invoked the fact that $\qtau, \qeta$ and $\sigma$ are actually gauge-invariant quantities. Therefore the identities (57)-(59) must hold on these quantities when derived from \emph{any} solution to the linearised Einstein equations on Schwarzschild and therefore in particular $\gamma$.

Finally to complete the proof of the main result we must show that $j$ and $L$ can be expressed as claimed in Proposition \ref{propintro} in terms of two scalars satisfying the Regge--Wheeler and Zerilli equations respectively. Indeed, in view of the previous comment regarding gauge-invariance of $\qtau, \qeta$ and $\sigma$ this statement on $j$ and $L$ then follows immediately from section 6 of \cite{Jlinstabschwarz}. In particular we have:
\begin{proposition}\label{propRWZ}
Let $\gamma$ be as in Proposition \ref{mainprop}. Then the derived quantities $\qtau, \qeta$ and $\sigma$ can be expressed as
\begin{align*}
\qtau&=\qn\astrosunhat\qexd\Big(r\Psi\Big)+6\mu\qexd r\qastrosunhat\zslapinv{1}\qexd\Psi,\\
\qeta&=-\qhd\qexd\Big(r\Phi\Big),\\
\sigma&=-2r\slap\Psi+4\qn_{\qP}\Psi+12\mu r^{-1}(1-\mu)\zslapinv{1}\Psi
\end{align*}
where the smooth scalars $\Phi$ and $\Psi$ satisfy the Regge--Wheeler and Zerilli equations respectively:
\begin{align*}
\qbox\Phi+\slap\Phi=-\frac{6}{r^2}\frac{M}{r}\Phi
\end{align*}
and
\begin{align*}
\qbox\Psi+\slap\Psi=-\frac{6}{r^2}\frac{M}{r}\Psi+\frac{24}{r^5}\frac{M}{r}(r-3M){\slap^{-1}_{\mathfrak{Z}}}\Psi+\frac{72}{r^7}\frac{M}{r}\frac{M}{r}(r-2M)\zslapinv{2}\Psi.
\end{align*}
Here $\zslapinv{p}$ is in the inverse of the operator $\slap+\frac{2}{r^2}(1-\frac{3M}{r})$ applied $p$-times.
\end{proposition}

\appendix

\section{An effective scalarisation of \eqref{leett}}

In this appendix we show how one can effectively scalarise \eqref{leett} using Proposition \ref{mainprop}.

\begin{proposition}\label{propFIeqnandeqnlingrav}
	Let $\gamma$ be as in Proposition \ref{mainprop} and define the smooth functions\footnote{Remark that $\rho=0$ if $A=\exd \chi$.}
	\begin{align*}
	\rho=r^2\qhd\qexd\qA,
	\end{align*}
and
\begin{align*}
\omega=r^3\slap\even{\sA},\qquad\omegab=r^2\slap\odd{\sA}.
\end{align*}
	Then the following system of equations hold true\footnote{Remark that $r\Boxgsch(r^{-1}f)=\qbox f+\slap f$ on smooth functions.}:
	\begin{align*}
	\qbox\rho+\slap\rho=-\qhd\qexd\bigg(r^2\qj\bigg)
	\end{align*}
	and
	\begin{align*}
	\qbox\omega+\slap\omega=\frac{\mu}{r^2}\omega-2\qhd\qn\rho+2r^2\qj,\qquad\qbox\omegab+\slap\omegab=-2r\slap\qeta_P.
	\end{align*}
\end{proposition}

\begin{proof}
	The equation on $\omegab$ follows from \eqref{cock2}. To derive the equation for $\rho$ we first introduce the quantity $\qxi=\qA-\qn\even{\sA}$ and then derive from \eqref{cock1}-\eqref{cock3} the equation
	\begin{align}\label{eqnform}
	-\frac{1}{r^2}\qdiv\Big(r^2\qexd\qxi\Big)+\slap\qxi=-\qj
	\end{align}
	which yields the desired equation on $\rho=r^2\qhd\qexd\qA= r^2\qhd\qexd\qxi$.
	
	For the last equation we derive from \eqref{cock2}
	\begin{align*}
	\qbox(r^{-1}\omega)+\slap(r^{-1}\omega)=\frac{2}{r}\qn_{\qP}(r^{-1}\omega)-2r\slap\qxi_{\qp}
	\end{align*}
	and then use \eqref{eqnform}.
\end{proof}
Now as promised we achieve the scalarisation of $\gamma$ in that one can express angular derivatives of $\gamma$ in terms of derivatives of the scalars $\rho, \omega$ and $\omegab$ of the above proposition and the invariant quantities $\Phi$ and $\Psi$ which determine $\qtau, \qeta$ and $\sigma$ (cf. Proposition \ref{propRWZ}). Estimates on this scalar hierarchy then translate easily to estimates on $\gamma$ after appealing to elliptic estimates on spheres.

\begin{proposition}\label{propunravellingthehierarchy}
	Let $\gamma$ be as in Proposition \ref{propFIeqnandeqnlingrav} and define the quantity
	\begin{align*}
	\qZ:=-2r^2\qj+\qhd\qn\rho+\qn\Big(r^{-1}\omega\Big).
	\end{align*}
	Then the following relations hold:
	\begin{align*}
	r^2\slap\qhatgamma&=r^2\slap\widehat{\qtau}+\qn\qastrosunhat\qZ,\\
r^2\slap \qtrgamma&=2\qdiv\qZ,\\
	r^2\sdso\sdo{\mgamma}&=\sdso\Big(\qZ, r^2\slap\qeta\Big)+r^2\sdso\Big(\qn\big(r^{-3}\omega\big), \qn\big(r^{-2}\omegab\big)\Big),\\
	r^2\sdst\sdt\shatgamma&=\sdst\sdso\Big(r^{-1}\omega, \omegab\Big),\\	r^2\slap\strgamma&=r^2\slap\sigma+\frac{4}{r}\qZ_P+\frac{2}{r}\slap\omega.
	\end{align*}
\end{proposition}
\begin{proof}
	Applying the relevant angular operators to the derived quantities we find
	\begin{align*}
	r^2\slap\qhatgamma&=r^2\slap\widehat{\qtau}+\qn\qastrosunhat\Big( r^2\slap\qzeta\Big),\\
r^2\slap\qtrgamma&=2\qdiv \Big(r^2\slap\qzeta\Big),\\
r^2\sdso\sdo{\mgamma}&=\sdso\Big( r^2\slap\qzeta, r^2\slap\qeta\Big)+r^2\sdso\Big(\qexd\big(r^{-3}\omega\big), \qexd\big(r^{-2}\omegab\big)\Big),\\
r^2\sdst\sdt\shatgamma&=\sdst\sdso\Big(r^{-1}\omega, \omegab\Big),\\
r^2\slap\strgamma&=r^2\slap\sigma+4r\slap\qzeta_{\qP}+\frac{2}{r}\slap\omega.
	\end{align*}
The proposition then follows after noting that $\qZ= r^2\slap\qzeta$ which in turn follows from \eqref{eqnform} and the fact that $r^2\slap\qzeta=r^2\slap\qxi+\qn(r^{-1}\omega)$.
	
\end{proof}

\end{document}